\documentclass[11pt,a4paper]{article}
\usepackage{fullpage}

\usepackage[utf8]{inputenc}
\usepackage{amsthm,amsmath,amssymb}
\usepackage[mathcal,mathscr]{eucal}
\usepackage{epsfig,graphicx,graphics,color}
\usepackage{enumerate}
\usepackage[sort,nocompress]{cite}
\usepackage{hyperref}
\hypersetup{
    colorlinks=true,       
    linkcolor=blue,        
    citecolor=red,         
    filecolor=magenta,     
    urlcolor=cyan,         
    linktocpage=true
}
\usepackage{thmtools}
\usepackage{subcaption,thm-restate}

\usepackage{algorithm}
\usepackage{bbm}
\usepackage[noend]{algpseudocode}
\usepackage{varwidth}
\newcounter{algsubstate}


\newcommand{\pmax}{p_{\max}}

\newcommand{\jobs}{J}

\newcommand{\opt}{\mathsf{opt}}
\newcommand{\spt}{\mathsf{spt}}
\newcommand{\opv}{{\rm opt}}
\newcommand{\spv}{{\rm spt}}

\newcommand{\Nset}{\ensuremath{\mathbb N}}

\def\final{0}  
\ifnum\final=0  
\newcommand{\knote}[1]{{\color{red}[{\tiny Krist\'of: \bf #1}]\marginpar{\color{red}*}}}
\newcommand{\snote}[1]{{\color{blue}[{\tiny Simon: \bf #1}]\marginpar{\color{blue}*}}}
\newcommand{\mnote}[1]{{\color{green}[{\tiny Matthias: \bf #1}]\marginpar{\color{green}*}}}
\newcommand{\tnote}[1]{{\color{cyan}[{\tiny Tam\'as: \bf #1}]\marginpar{\color{cyan}*}}}
\else 
\newcommand{\knote}[1]{}
\newcommand{\snote}[1]{}
\newcommand{\mnote}[1]{}
\newcommand{\tnote}[1]{}
\newcommand{\todo}[1]{}
\fi

\theoremstyle{plain}
\newtheorem{theorem}{Theorem}

\newtheorem{claim}[theorem]{Claim}

\theoremstyle{definition}

\title{Scheduling with Non-Renewable Resources:\\ Minimizing the Sum of Completion Times\thanks{ Supported by DAAD with funds of the Bundesministerium f{\"u}r Bildung und Forschung (BMBF) and by DFG project MN 59/4-1.}}

\author{Krist{\'o}f B{\'e}rczi\thanks{MTA-ELTE Egerv\'ary Research Group, Department of Operations Research, E{\"o}tv{\"o}s Lor{\'a}nd University, Budapest, Hungary. Email: \texttt{berkri@cs.elte.hu}.}
   \and Tam{\'a}s Kir{\'a}ly\thanks{MTA-ELTE Egerv\'ary Research Group, Department of Operations Research, E{\"o}tv{\"o}s Lor{\'a}nd University, Budapest, Hungary. Email: \texttt{tkiraly@cs.elte.hu}.}
   \and Simon Omlor\thanks{TU Hamburg, Institute for Algorithms and Complexity, Hamburg, Germany. Email: \texttt{simon.omlor@tuhh.de}.}}

\begin{document}
\maketitle
\begin{abstract}
  The paper considers single-machine scheduling problems with a non-renewable resource.
  In this setting, we are given a set jobs, each of which is characterized by a processing time, a weight, and the job also has some resource requirement.
  At fixed points in time, a certain amount of the resource is made available to be consumed by the jobs.
  The goal is to assign the jobs non-preemptively to time slots on the machine, so that at any time their resource requirement does not exceed the available amounts of resources.
  The objective that we consider here is the minimization of the sum of weighted completion times.
  
  We give polynomial approximation algorithms and complexity results for single scheduling machine problems. 
  In particular, we show strong NP-hardness of the case of unit resource requirements and weights ($1|rm=1,a_j=1|\sum C_j$), thus answering an open question of Gy\"orgyi and Kis. 
  We also prove that the schedule corresponding to the Shortest Processing Time First ordering provides a $3/2$-approximation for the same problem. 
  We give simple constant factor approximations and a more complicated PTAS for the case of $0$ processing times ($1|rm=1,p_j=0|\sum w_jC_j$).
  We close the paper by proposing a new variant of the problem in which the resource arrival times are unknown.
  A $4$-approximation is presented for this variant, together with an $(4-\varepsilon)$-inapproximability result.

\medskip

\end{abstract}

\raisebox{-65ex}[0pt][0pt]{\hspace{78ex}\includegraphics[scale=0.45]{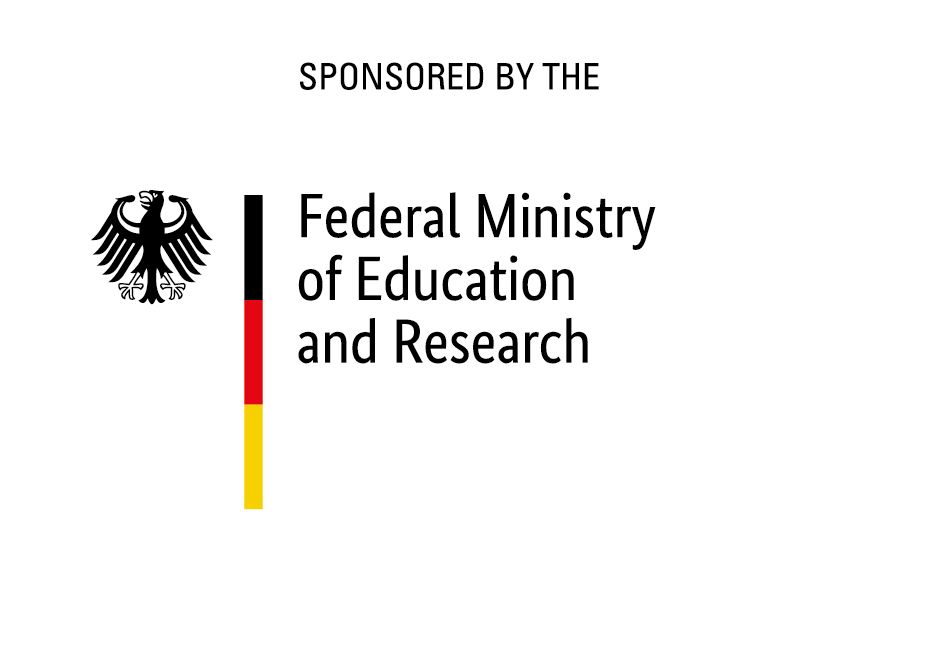}}

\section{Introduction}
\label{sec:introduction}

The problem of scheduling with non-renewable resources appears naturally in practical problems where resources like raw materials, energy, or financial constraints are taken into account. 
These problems are interesting both from the practical and from the theoretical point of view. In the general setting, we are given a set of jobs and a set of machines. 
Each job is equipped with a requirement vector that encodes the needs of the given job for the different types of resources. 
There is an initial stock for each resource, and some additional resource arrival times in the future are known together with the arriving quantities. 
The aim is to find a schedule of the jobs on the machines such that the resource requirements are met.

We will use the standard $\alpha|\beta|\gamma$ notation of Graham, Lawler, Lenstra and Kan \cite{graham1979optimization}. Grigoriev, Holthuijsen and Van De Klundert \cite{grigoriev2005basic} extended this notation by adding the restriction $rm=r$ to the $\beta$ field, meaning that there are $r$ resources (raw materials). 
In the present paper, we concentrate on problem $1|rm=1|\sum w_jC_j$, that is, when we have a single machine, a single resource, and the goal is to minimize the weighted sum of completion times. While there is a long list of results on the approximability of the makespan objective, much less is known about the complexity and approximability of the total weighted completion time objective.

\paragraph{Previous work}
Scheduling problems with resource restrictions were introduced by Carlier \cite{carlier1984problemes} and Slowinski~\cite{slowinski1984preemptive}. 
Carlier settled the computational complexity of several variants for the sinlge machine case~\cite{carlier1984problemes}. In particular, it was shown that $1|rm=1|\sum w_jC_j$ is NP-hard in the strong sense. This was also proved independently by Gafarov, Lazarev and Wener in \cite{gafarov2011single}.
Kis~\cite{kis2015approximability} showed that the problem remains weakly NP-hard even when the number of resource arrival times is $2$. 
On the positive side, he gave an FPTAS for $1|rm=1,q=2|\sum w_jC_j$. 
A variant of the problem where each job has processing time $1$, there are $q=n$ resource arrival times such that $t_i=iM$ and $b_i=M$ for $i=1,\dots,n$, and $M=\sum_{j\in J} a_j/n$ is an integer, was considered in \cite{gafarov2011single}. 
Recently, Gy\"orgyi and Kis~\cite{gyorgyi2018minimizing} gave polynomial time algorithms for several special cases, and also showed that the problem remains weakly NP-hard even under the very strong assumption that the processing time, the resource requirement and the weight are the same for each job. They also provided a $2$-approximation algorithm for this variant, and a polynomial-time approximation scheme (PTAS) when the number of resource arrival times is a constant and the processing time equals the weight for each job, while the resource requirements are arbitrary.

In contrast to the case of total weighted completion objective, much is known about scheduling problems with non-renewable resources for the maximum makespan and maximum lateness objectives. 
Slowinski~\cite{slowinski1984preemptive} studied the preemptive scheduling of independent jobs on parallel unrelated machines with the use of additional renewable and non-renewable resources under financial constraints. 
Toker, Kondakci and Erkip~\cite{toker1991scheduling} examined a single machine scheduling problem under non-renewable resource constraint, using the makespan as a performance criterion.
Xie~\cite{xie1997polynomial} generalized this result to the problem with multiple financial resource constraints. 
Grigoriev, Holthuijsen and Van De Klundert~\cite{grigoriev2005basic} presented polynomial time algorithms, approximations and complexity results for single scheduling machine problems with unit or all equal processing times, and maximum lateness and makespan objectives. 
In a series of papers~\cite{gyorgyi2014approximation,gyorgyi2015reductions,gyorgyi2015approximability,gyorgyi2017approximation,gyorgyi2017ptas}, Gy\"orgyi and Kis presented approximation schemes and inapproximability results both for single and parallel machine problems with the makespan and the maximum lateness objectives. 
In \cite{gyorgyi2018minimizing}, they proposed a bracnh-and-cut algorithm for minimizing the maximum lateness.

\paragraph{Our results}

The first problem that we consider is $1|rm=1,a_j=1|\sum C_j$. The complexity of this problem was posed as an open question in~\cite{gyorgyi2018minimizing}. We show that the problem is NP-hard in the strong sense. 

\begin{restatable}{theorem}{thmsnp}
\label{thm:snp}
$1|rm=1,a_j=1|\sum C_j$ is strongly NP-hard.
\end{restatable}

In the light of Theorem~\ref{thm:snp}, one might be interested in finding an approximation algorithm for the problem. Given any scheduling problem on a single machine, the \textbf{Shortest Processing Time First} (SPT) schedule orders the jobs by processing times, i.e. $p_{\spt^{-1}(i)}\leq p_{\spt^{-1}(i+1)}$ for all $i$. We prove that $\spt$ provides a $3/2$-approximation. Although the algorithm is merely scheduling according to the SPT order, the analysis if the algorithm is rather involved.

\begin{restatable}{theorem}{thmthreehalf}
\label{thm:32}
The SPT schedule gives a $\frac{3}{2}$-approximation for $1|rm=1,a_j=1|\sum C_j$, and the approximation guarantee is tight.
\end{restatable}

The second problem considered is the special case when the processing time is $0$ for every job. This setting is a relaxation of those instances where the processing times are short and the resource arrival times are far away from each other. First we give a $6$-approximation based on a non-trivial greedy approach.

\begin{restatable}{theorem}{thmgreedy}
\label{thm:greedy}
There exists a 6-approximation for $1|rm=1,p_j=0|\sum C_jw_j$ with running time $\mathcal{O}(n \log(n))$.
\end{restatable}

We give a slightly more complicated $(4+\varepsilon)$-approximation that illustrates one of the important ideas of the general PTAS.

\begin{restatable}{theorem}{thmfourpluseps}
\label{thm:4+eps}
There exists a $(4+\varepsilon)$-approximation for $1|rm=1,p_j=0|\sum C_jw_j$ with running time polynomial in $1/\varepsilon$ and the input length.
\end{restatable}

As a next step toward an efficient approximation algorithm, we present a PTAS for the case of a constant number of resource arrival times. This procedure will be used as a subroutine in our algorithm for the general case.

\begin{restatable}{theorem}{thmpreptas}
\label{thm:preptas}
There exists an $(1+\frac{q}{k})$-approximation for $1|rm=1,p_j=0|\sum C_jw_j$ with running time $\mathcal{O}(n^{qk+1})$.
\end{restatable}

Finally, we prove the main result of the paper, which is a PTAS for the case of an arbitrary number of resource arrival times.

\begin{restatable}{theorem}{thmptas}
\label{thm:ptas}
There exists a PTAS for $1|rm=1,p_j=0|\sum C_jw_j$.
\end{restatable}

The last problem that we is another variant of $1|rm=1,p_j=0|\sum C_jw_j$ where the arrival times are unknown. We denote this problem by $1|rm=1, p_j=0, \text{$t_i$ unkown}|\sum C_jw_j$.

\begin{restatable}{theorem}{thmunknown}
\label{thm:unknown}
There exists a $(4+\varepsilon)$-approximation for $1|rm=1, p_j=0, \text{$t_i$ unkown}|\sum C_jw_j$  with running time polynomial in $1/\varepsilon$ and the input length. Moreover, there is no $(4-\varepsilon)$-approximation algorithm for the problem for any $\varepsilon >0$.
\end{restatable}

\paragraph{Organization}

The rest of the paper is organized as follows. Basic notation and terminology are introduced in Section~\ref{sec:preliminaries}. A strong NP-hardness proof and a $3/2$-approximation algorithm for problem $1|rm=1,a_j=1|\sum C_j$ are given in Section~\ref{sec:1}. Results on problem $1|rm=1,p_j=0|\sum C_j$ are discussed in Section~\ref{sec:2}, where a greedy $6$-approximation, a $(4+\varepsilon)$-approximation, a PTAS for the case of constant resource arrival times, and a PTAS for the general case are presented. We close the paper by proposing a new variant of the problem in which the resource arrival times are unknown. A $4$-approximation is presented for this case, together with an $(4-\varepsilon)$-inapproximability result.

\section{Preliminaries}
\label{sec:preliminaries}

Throughout the paper, we will use the following notation. 
We are given 
a set $J$ of $n$ \textbf{jobs}. 
Each job $j\in J$ has a non-negative integer processing time $p_j$, a non-negative weight $w_j$, and a resource requirement $a_j$. 
The \textbf{resources arrive at time points $t_1,\dots,t_q$}, and the \textbf{amount of resource that arrives at time point $t_i$} is denoted by $b_i$.
We might assume that $\sum_{i=1}^q b_i = \sum_{j=1}^n a_j$ holds.
We will always assume that $t_1=0$, as this does not effect the approximation ratio of our algorithms.

The jobs should be processed non-preemptively on a single machine. A \textbf{schedule} is an ordering of the jobs, that is, a mapping $\sigma: J \rightarrow [n]$, where $\sigma(j)=i$ means that job $j$ is the $i$th job scheduled on the machine. 
The \textbf{completion time} of job $j$ in schedule $\sigma$ is denoted by $C^\sigma_j$. 
We will drop the index $\sigma$ if the schedule is clear from the context. 
In any reasonable schedule, there is an idle time before a job $j$ only if there is not enough resource left to start $j$ after finishing the last job before the idle period.
Hence the completion time of job $j$ is basically determined by the ordering and by the resource arrival times, as $j$ will be scheduled at the first moment when the preceding jobs are already finished and the amount of available resource is at least $a_j$. 


\section{The problem \texorpdfstring{$1|rm=1,a_j=1|\sum C_j$}{1|rm=1,aj=1|sum Cj}} \label{sec:1}

\subsection{Strong NP-completeness} \label{sec:np}

The aim of this section is to prove Theorem~\ref{thm:snp}.

\thmsnp*
\begin{proof}
Recall that all $a_j$ and $w_j$ values are 1, and each job has an integer processing time $p_j$.
The number of resource arrival times is part of the input.

We prove NP-completeness by reduction from the {\sc 3-Partition} problem. 
The input contains numbers $B \in \Nset$, $n \in \Nset$, and $x_j \in \Nset$ $(j=1,\dots,3n)$ such that $B/4 < x_j< B/2$ and $\sum_{j=1}^{3n} x_j=nB$.
A feasible solution is a partition $J_1,\dots,J_n$ of $[3n]$ such that $|J_i|=3$ and $\sum_{j \in J_i} x_j=B$ for every $i \in [n]$.
In contrast to the {\sc Partition} problem, the {\sc 3-partition} problem remains NP-complete even when the integers $x_j$ are bounded above by a polynomial in $n$. 
That is, the problem remains NP-complete even when the numbers in the input are represented  as unary numbers \cite[Pages 96–105 and 224]{gj}.

Let $K=4nB$.
The reduction to $1|rm=1,a_j=1|\sum C_j$ involves three types of jobs.
\begin{description}
  \item{\textbf{Normal jobs}} These correspond to the numbers $x_j$ in the {\scshape 3-Partition} instance, so there are $3n$ of them and the processing time $p_j$ of the $j$-th normal job is $x_j$.
  \item{\textbf{Small jobs}} Their processing time is 1 and there are $nK$ of them.
  \item{\textbf{Large jobs}} Their processing time is $K$ and there are $nK$ of them.
\end{description}

There are also three types of resource arrivals:
\begin{description}
  \item{\textbf{Type 1}} Three resources arrive at times $i(B+K)$ ($i=0,\dots,n-1$).
  \item{\textbf{Type 2}} One resource arrives at $i(B+K)+j$\\ ($i=0,\dots,n-1$, $j=B,\dots, B+K-1$).
  \item{\textbf{Type 3}} One resource arrives at $n(B+K)+iK$ ($i=0,\dots,nK-1)$.
\end{description}

\begin{figure}[h!]
\centering
\includegraphics[width=\linewidth]{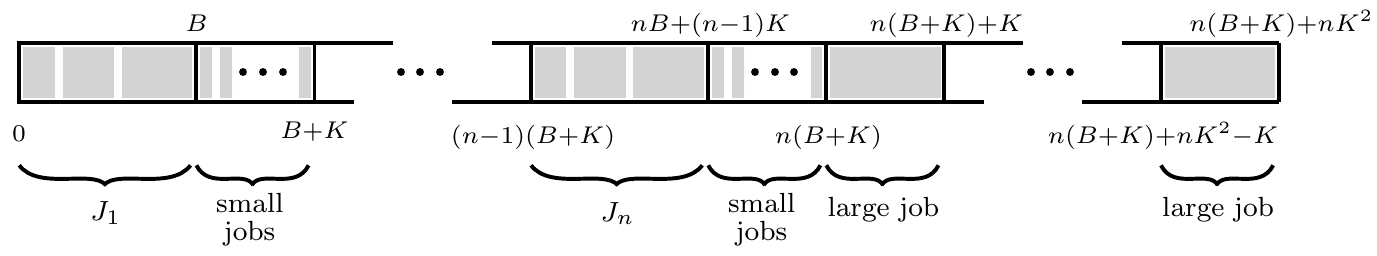}
\caption{The schedule corresponding to a feasible solution of {\sc 3-Partition}.}
\label{fig:3partred}
\end{figure}

Suppose that the {\sc 3-Partition} instance has a feasible solution $J_1,\dots,J_n$.
We consider the following schedule $S$: resources of Type 1 are used by normal jobs, such that jobs in $J_i$ are scheduled between $(i-1)(B+K)$ and $iB+(i-1)K$ (in $\spt$ order).
Type 2 resources are used by small jobs that start immediately.
Type 3 resources are used by the large jobs that also start immediately at the resource arrival times (see Figure~\ref{fig:3partred}).

Instead of $\sum C_j$, consider the equivalent shifted objective function $\sum (C_j-t_j-p_j)$, where $t_j$ is the arrival time of the resource used by job $j$ and $p_j$ is its processing time -- we assume without loss of generality that resources are used by jobs in order of arrival.
Note that all terms of $\sum (C_j-t_j-p_j)$ are nonnegative.
As small jobs and large jobs start immediately at the arrival of the corresponding resource in schedule $S$, their contribution to the shifted objective function is $0$.
The jobs in $J_i$ have total processing time $B$, and their contribution to the shifted objective function is two times the processing time of the shortest job plus the processing time of the second shortest job, which is at most $B$. 
Hence the schedule $S$ has objective value at most $nB$.

We claim that if the {\sc 3-Partition} instance has no feasible solution, then the objective value of any schedule is strictly larger than $nB$.
First, notice that if a large job is scheduled to start before time $n(B+K)$, then $\sum (C_j-t_j-p_j)$ has a term strictly larger than $nB$ as there is a resource that arrives while the large job is processed and is not used for more than $nB$ time units.
Similarly, if the first large job starts at $n(B+K)$ but uses a resource that arrived earlier, then the resource that arrives at $n(B+K)$ is not used for more than $nB$ time units. We can conclude that the first large job uses the resource arriving at $n(B+K)$. 

If the first large job does not start at $n(B+K)$, then all large jobs have positive contribution to the objective value, so again, the objective value is larger than $nB$.
We can therefore assume that the large jobs start exactly at $n(B+K)+iK$ ($i=0,\dots,nK-1)$ and that there is no idle time before $(B+K)n$. In particular this means all other jobs are already completed at $(B+K)n$.

Consider Type 2 resources arriving at $i(B+K)+j$ ($j=B,\dots, B+K-1$) for some fixed $i$.
If the first or the second resource is not used immediately, then none of the subsequent ones are, so the objective value is more than $nB$.
Hence, the first resource must be used immediately by a small job.

Suppose that some resource in this interval is used by a normal job.
If it is followed by a small job, then we may improve the objective value by exchanging the two. Thus, in this case, we can assume that the last resource of the interval is used by a normal job, and also 
the Type 1 resources arriving at $(i+1)(B+K)$ are used by normal jobs. But this is impossible, because normal jobs have processing time at least $B/4+1$, and a small job starts at time $(i+1)(B+K)+B$.

To sum up, we can assume that all resources of Type 2 are used immediately by small jobs.
This means that normal jobs have to use resources of Type 1, and must exactly fill the gaps of length $B$ between the arrival of resources of Type 2.
This is only possible if the 3-partition instance has a feasible solution, concluding the proof of Theorem~\ref{thm:snp}.
\end{proof}

\subsection{Shortest processing time first for unit resource requirements} \label{sec:spt}

In the previous section, we have seen that scheduling with a non-renewable resource is strongly NP-hard already for unit resource requirements.
Now we show that scheduling the jobs according to an $\spt$ ordering provides a $3/2$-approximation for the problem with unit weight and unit resource requirements, thus proving Theorem~\ref{thm:32}.

\thmthreehalf*
\begin{proof}

  To prove the theorem consider any instance $I$.
  We denote the completion times for the $\spt$ ordering by $C_j$ and their sum by $\spv$. 
  Furthermore, let $S_{\spt^{-1}(i)}:=C_{\spt^{-1}(i)}-p_{\spt^{-1}(i)}$ denote the starting time of the $i$th job in the $\spt$ schedule.
  Let $\opt$ be the optimal schedule for $I$.
  We denote the completion times for $\opt$ by $C_j'$ and their sum by $\opv$.
  Furthermore, let $S_{\opt^{-1}(i)}':=C_{\opt^{-1}(i)}'-p_{\opt^{-1}(i)}$ denote the starting time of the $i$th job in the optimal schedule.
  
  Our strategy is to simplify the instance by revealing its structural properties while not decreasing $\frac{\spv}{\opv}$. This way we get an upper bound for the approximation factor.
  We first consider the resource arrival times.

  \begin{claim}
    We may assume that the $i$th resource arrives at $S_{\opt^{-1}(i)}'$ for $i=1,\dots,n$.
  \end{claim}
  \begin{proof}
  As the $i$th resource is used by job $\opt^{-1}(i)$, the arrival time of that resource is at most $S_{\opt^{-1}(i)}'$. 
  If we move the arrival time of the resource to exactly $S_{\opt^{-1}(i)}'$, then $\opt$ does not change and $\spt$ cannot decrease.
  \end{proof}
  
The next claim shows that we can get rid of the idle times in the optimal schedule.

  \begin{claim}\label{claim:noidle}
    We may assume that there is no idle time in schedule $\opt$, that is, $S_{\opt^{-1}(i)}'=C_{\opt^{-1}(i-1)}'$ for $i=2,\dots,n$.
  \end{claim}
  \begin{proof}
    Suppose that there is some $i$ such that $t_i> C'_{ \opt^{-1}(i-1)}$.
    We reduce $t(i')$ by $\Delta=t(i)- C'_{ \opt^{-1}(i-1)}$ for all $i'\geq i$.
    Then for each $i'\geq i$, the completion time $C_{\opt^{-1}(i')}'$ decreases by $\Delta$.
    For each $i'\geq i$, the completion time $C_{\spt^{-1}(i')}$ decreases by at most $\Delta$.
    This follows from the fact that the resource arrival times decrease by $\Delta$ and the completion time of the previous job can decrease by at most $\Delta$ (which can be shown by induction).
    Hence $\opv$ decreases by at least as much as $\spv$.
    Since $\spv\geq \opv$, the ratio $\frac{\spv}{\opv}$ will not decrease by this change.
  \end{proof}

  Next, we modify the processing times.

  \begin{claim}
    We may assume that $p_{\opt^{-1}(1)}>p_{\spt^{-1}(1)}$ and that $p_{\spt^{-1}(1)}=0$.
  \end{claim}
  \begin{proof}
    If both schedules start with the same job, then we can remove the job from the instance and decrease $b_1$ by $1$.
    Then $\opv$ decreases by the same amount as $\spv$.
    We can repeat this until the schedules start with jobs of different processing times.
    Now $p_{\opt^{-1}(1)}>p_{\spt^{-1}(1)}$, since $\spt$ starts with the shortest job.
    Decreasing the processing time of job $\spt^{-1}(1)$ to $0$ (without changing any arrival time) decreases $\spv$ by $p_{\spt^{-1}(1)}$ and $\opv$ by at least $p_{\spt^{-1}(1)}$. We can eliminate idle times in the new optimal schedule as in the proof of Claim \ref{claim:noidle}.
  \end{proof}

  \begin{claim}
    We may assume that $p_j \in \{0, 1\}$ for all $j \in J$.
  \end{claim}
  \begin{proof}
    Let $\pmax=\max_{j\in J}p_j$ be the maximum processing time.
    Scaling the processing times by dividing all processing and arrival times by $\pmax$ has no effect on $\frac{\spv}{\opv}$, hence we may assume that $\pmax=1$.
    Now assume that there is a job $j'$ with $p=p_{j'} \in (0, 1)$.
    Let $p_+=\min\{p_j ~|~ j\in J, p_j>p\}$ and $p_-=\max\{p_j ~|~ j\in J, p_j<p\}$.
    Let $J_p=\{j \in J~|~ p_j=p\}$ be the set of jobs with processing time $p$.
    We will show that we can either increase the processing time of all jobs in $J_p$ to $p_+$ or decrease the processing time of all jobs in $J_p$ to $p_-$ without decreasing $\frac{\spv}{\opv}$.

    For $j \in J$, let $h_j$ denote the number of jobs processed after $j$ plus $1$ with respect to schedule $\opt$, i.e. $h_j=n-\opt^{-1}(j)+1$.
    Changing the processing times of all jobs in $J_p$ by some $\Delta \in [p_- -p, p_+-p]$ and adopting the arrival times of the resources will increase $\opv$ by $\Delta\sum_{j \in J_p}h_j$. Indeed, every time we change the processing time of one job $j$, the completion time of $j$ and of all jobs after $j$ will be increased by $\Delta$. Note that $\Delta$ can be negative, which means the completion times can also 'decrease' by $|\Delta|$.
    
%
%
    
    Notice that the order of the jobs in the $\spt$ schedule is not changed.
    Consider the $\spt$ schedule before the change.
    Let job $j\in J$ be any job, and let $j_0$ be the the first job that is processed after the last idle time before the starting time of $j$.
    Let $f_j$ be the number of jobs $j' \in J_p$ with $C_{j'}'\leq S_{j_0}$.
    For each of those jobs, the arrival time of the resource needed to start $j_0$ will be changed by $\Delta f_j$.
    Thus, the starting time of job $j_0$ in the changed $\spt$ schedule is at least $S_{j_0}+ \Delta f_j$.
    Now let $g_j$ be the number of jobs $j' \in J_p$ that are processed in the time interval $[S_{j_0}, C_j)$ before the change.
    For each of those jobs, the processing time is changed by $\Delta$ and the job is started at or after the new starting time of job $j_0$.
    Thus the new completion time of $j$ is at least $C_j+ \Delta f_j + \Delta g_j$.
    Consequently, $\spv$ will increase by at least $\sum_{j \in J}(f_j+g_j)\Delta$ if $\Delta>0$, and decrease by at most $\sum_{j \in J}(f_j+g_j)|\Delta|$ if $\Delta<0$.
    
    If $\frac{\sum_{j \in J}(f_j+g_j)}{\sum_{j \in J_p}h_j}\geq \frac{\spv}{\opv}$, then increasing the processing times in $J_p$ to $p_+$ will not decrease $\frac{\spv}{\opv}$.
    Otherwise, decreasing the processing times in $J_p$ to $p_-$ will not decrease $\frac{\spv}{\opv}$.
    Each time we apply this operation, the number of distinct processing times decreases by $1$. Finally, we get an instance where the only processing times are $p_{\min}=0$ and $p_{\max}=1$.
  \end{proof}

  Lastly, we modify the order of the jobs in the optimal solution.
  If $\opt$ and $\spt$ process a job of length $0$ at the same time, then we can remove the job from the instance and reduce the number of resources that arrive at this time by $1$.
  This will reduce $\opv$ and $\spv$ by the same amount. 
  
  Let $t$ be the time at which schedule $\spt$ first starts to process a job of length $1$. 
  On one hand, $\opt$ does not process jobs of length $0$ before $t$ by the above argument. On the other hand, 
  there is no idle time after $t$ in $\spt$, because that would mean idle time in $\opt$.
  Thus, if we move all jobs of length $0$ and their corresponding resource arrivals in $\opt$ to time $t$, then $\spv$ does not change but $\opv$ decreases. We may thus assume that schedule $\opt$ processes every job of length $0$ at $t$.
  
 \begin{figure}[h!]
\centering
\includegraphics[width=0.8\linewidth]{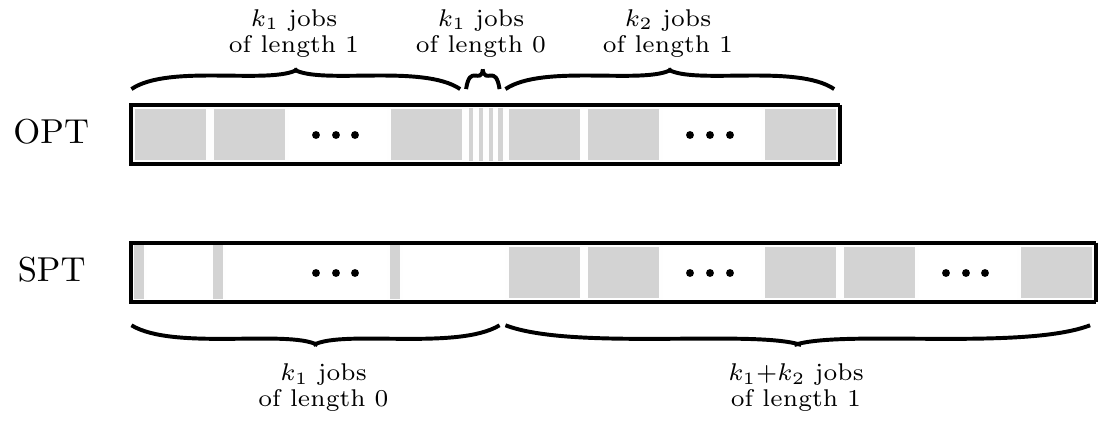}
\caption{Schedules $\opt$ and $\spt$ after the reductions. The $0$ length jobs are scheduled in $\opt$ to the first resource arrival time when multiple resources arrive.}
\label{fig:worstcase}
\end{figure}
  
  We conclude that $\opt$ first processes $k_1$ jobs of length $1$, then $k_1$ jobs of length $0$ and then $k_2$ jobs of length $1$, while $\spt$ starts with the jobs of length $0$ having a lot of idle time in the beginning and then consecutively processes all jobs of length $1$ (see Figure~\ref{fig:worstcase}).
  The weighted sums of completion times are then given by
  \begin{equation*}
    \opv=\frac{k_1(k_1+1)}{2}+k_1^2+k_2k_1+\frac{k_2(k_2+1)}{2}
  \end{equation*}
  and
  \begin{equation*}
    \spv=\frac{k_1(k_1-1)}{2}+k_2k_1+\frac{k_1(k_1+1)}{2}+(k_1+k_2)k_1+\frac{k_2(k_2+1)}{2}  .
  \end{equation*}
  We get
  \begin{equation*}
    \frac{3}{2}\opv -\spv=\frac{k_1^2}{4}+\frac{k_2^2}{4}-\frac{k_1k_2}{2}+\frac{3k_1+k_2}{4} \geq \frac{(k_1-k_2)^2}{4}\geq 0  ,  
  \end{equation*}
showing that the approximation factor is at most $\frac{3}{2}$. 

Setting $k_2=k_1$ and letting $k_1$ go to infinity gives us a sequence of instances such that $\frac{\spv}{\opv}$ converges to $\frac{3}{2}$ as we have $\spv=\frac{9}{2}k_1^2+\mathcal{O}(k_1)$ and $\opv=3k_1^2+\mathcal{O}(k_1)$.
This concludes the proof of Theorem~\ref{thm:32}.
\end{proof}

\section{The problem \texorpdfstring{$1|rm=1,p_j=0|\sum C_jw_j$}{1|rm=1,pj=0|sum Cjwj}} \label{sec:2}

In this section we consider problem $1|rm=1,p_j=0|\sum C_jw_j$, another special case of $1|rm=1|\sum C_jw_j$.
The  problem clearly is NP-hard even for $q=2$ as the knapsack problem can be reduced to it. Indeed, maximizing the weight of the items in the knapsack is equivalent to the task of maximizing the weight of jobs that are scheduled at the first resource arrival time.
Recall that Kis~\cite{kis2015approximability} gave a FPTAS for $1|rm=1|\sum C_jw_j$ when there are two resource arrival times.

First we give a $6$-approximation for the problem based on a greedy approach. 
Next, we give a more complicated $4$-approximation that illustrates one of the important ideas of the more 
general PTAS.
Then we provide a PTAS for the case when $q$, the number of resource arrival times is a constant. This algorithm will be used as a subroutine in the PTAS for the general case.
Finally, we prove the main result of the paper which is a PTAS for the case of an arbitrary number of resource arrival times. 

Since the processing times are $0$, every job is processed at one of the arrival times in any optimal schedule.
Thus a schedule can be represented by a mapping $\pi: J \rightarrow [q]$, where $\pi(j)$ denotes the index of the resource arrival time when job $j$ is processed.
A schedule is feasible if the resource requirements are met, that is, if
\begin{equation}\label{eq:1}
\sum_{j: \pi(j)\leq k} a_j \leq \sum_{i \leq k} b_i
\end{equation}
for all $1\leq k\leq q$.
As we assume that $\sum_i b_i = \sum_j a_j$ holds, this is equivalent to
\begin{equation}\label{eq:2}
\sum_{j: \pi(j)\geq k} a_j \geq \sum_{i \geq k} b_i
\end{equation}
for all $1\leq k\leq q$.

Define $B_k=\sum_{i \geq k} b_i$, and consider the set of jobs that are not processed before a given time point $t_i$.
Then \eqref{eq:2} says that if the resource requirements of these jobs add up to at least $B_i$, then our schedule is feasible.
We will mostly use the second characterization of feasibility, as our algorithms assign the jobs to later time points first. 
The intuition is that the total weight $W_i$ of jobs that are not processed at a time point $t_i$ gives a lower bound for their contribution to the objective function. Therefore it is better to approximate $W_i$, rather than the weight of the jobs already processed, given by $\left(\sum_{j\in J}w_j\right)- W_i$.

\subsection{A greedy \texorpdfstring{$6$}{6}-approximation for arbitrary \texorpdfstring{$q$}{q}} \label{sec:greedy}

The idea of our first algorithm is to have a balance between adding jobs that have small weights and jobs that have high resource requirements.
More precisely, we will assign jobs to the time points going back in time. When we add a job to the set of jobs scheduled after a given time point, we will choose the most inefficient job, i.e. the job minimizing $w_j/a_j$ among all jobs that have weight at most the weight $W$ of all jobs that have already been chosen up to this point.
If there is no job with weight at most $W$, then we simply choose a job with minimal weight.
Intuitively, this rule guarantees that the jobs we choose are not too efficient but their total weight is not too large either.

\begin{algorithm}[h!]
  \caption{Greedy algorithm for $1|rm=1,p_j=0|\sum C_jw_j$.}\label{alg:greedy}
  \begin{algorithmic}[1]
    \Statex \textbf{Input:} Jobs $\jobs$ with $|\jobs|=n$, resource requirements $a_j$, weights $w_j$, resource arrival times $t_1\le  \ldots \le t_q$ and resource quantities $b_1, \dots b_q$.
    \Statex \textbf{Output:} A feasible schedule $\pi$. 
    \State Set $A=0$.
    \State Set $W=0$.
	\For{$i$ from $0$ to $q-1$}
		\While{$A< B_{q-i}$} 
			\If{there is an unassigned job $j$ with $w_j\leq W$}
				\State \begin{varwidth}[t]{0.85\linewidth}
				Let $j$ be an unassigned job with $w_j\leq W$ minimizing $w_j/a_j$.
				\end{varwidth}
			\Else
    			\State Let $j$ be an unassigned job minimizing $w_j$.
			\EndIf
			\State $W\leftarrow W+w_j$
			\State $A\leftarrow A+a_j$
			\State Set $\pi(j)=q-i$.
		\EndWhile
	\EndFor
    \State \textbf{return} $\pi $
  \end{algorithmic}
\end{algorithm}

\thmgreedy*
\begin{proof}
We claim that Algorithm~\ref{alg:greedy} satisfies the requirements of the theorem. 
To prove this, we show that at the end of each iteration of the \emph{for} loop, $W$ is a 6-approximation for the problem of finding a minimum weight job set $S$ consuming at least $B_i$ resources. 
Let $\opv$ be the optimum value of this problem.
Since the $B_i$'s are monotone decreasing and the algorithm does not consider the $B_i$'s when picking the jobs, it is enough to show that the statement holds for the case $q=2,i=0$.

Consider the last iteration of the \emph{while} loop when $W<\opv$ holds at the beginning of the loop.
Since there exists a set of jobs of total weight $\opv$, the value of $W$ at the end of this iteration is bounded by $2\cdot\opv$.
After this iteration, our algorithm will always pick jobs that are at least as inefficient as any of the jobs picked by the optimum solution.
Consequently, the \emph{while} loop will end at the latest when the weight of the jobs scheduled after this iteration is at least $\opv$.

Now consider the last iteration of the \emph{while} loop. At the beginning of the iteration, we have $A<B_{q-i}$ and $W$ is bounded by $3\cdot\opv$ by the above.
Since there exists a set of jobs of total weight $\opv$, $W$ is at most doubled during the iteration.
This means that $W$ is bounded by $6\cdot\opv$ at the end of the \emph{while} loop, as stated.

The running time follows by ordering the jobs according to their weight and by using AVL trees for picking $j$ in the \emph{while} loop.
\end{proof}

The following example shows that the bound is tight. We have $5$ jobs with weights $w_1=w_2=1-2\varepsilon$, $w_3=1-\varepsilon$, $w_4=1$ and $w_5=3$. The resource requirements are $a_1=a_2=\varepsilon/5$, $a_3=1-\varepsilon/2$, $a_4=1$ and $a_5=4$. The resource arrival times are $t_1=0$ and $t_2=1$, with resource quantities $b_1=5-\varepsilon/10$ and $b_2=1$.
Here the optimum solution is to schedule the job with weight $1$ to time point $t_2$ and all the remaining jobs to time point $t_1$. However, our algorithm will schedule the job with weight $1$ to time point $t_1$ and all the remaining jobs to $t_2$.

\subsection{A \texorpdfstring{$(4+\varepsilon)$}{4+epsilon}-approximation for arbitrary \texorpdfstring{$q$}{q}} \label{sec:4+eps}

Now we give a slightly better approximation for the problem. The algorithm is slightly more complicated than the one presented in Section~\ref{sec:greedy}, but the proof illustrates one of the important ideas of the general PTAS.

\thmfourpluseps*
\begin{proof}
The idea of the algorithm is as follows. We may assume without loss of generality that resource arrival times are integer.
First we shift all resource arrival times to powers of $2$.
For each time point $t=t_i$ in the shifted instance, we apply the FPTAS by Kis~\cite{kis2015approximability} to the instance which has only two resource arrival times $t_1$ and $t$, and the resource quantity for $t$ is $B_i$.
Denote the set of jobs assigned to $t$ this way by $S_i$.
Then, going back from the last time point $t_q$ to the first one $t_1$, we assign all unassigned jobs from $S_i$ to $t_i$, i.e. $\pi(j)=\max\{ i :\  j\in S_i \}$.

More formally, let $\mathcal{I}$ be an instance of $1|rm=1, p_j=0 |\sum_j C_jw_j $. We assume $t_1=0$ and $t_2=1$.
We define a new instance $\mathcal{I}'$ of $1|rm=1, p_j=0 |\sum_j C_jw_j $ with shifted resource arrival times as follows. 
Set
\begin{equation*}
    t'_i=
    \begin{cases}
    0 &\quad \text{if $i=1$},\\
    2^{i-2} &\quad \text{for $i=2,\dots,\lceil \log_2(t_q) \rceil +2$},
    \end{cases}
\end{equation*}
and
\begin{equation*}
    b'_i=
    \begin{cases}
    b_i & \quad\text{if $i=1,2$},\\
    \sum[b_i:\ t_i\in (2^{i-3}, 2^{i-2}]\ ] &\quad \text{for $i=3,\dots,\lceil \log_2(t_q) \rceil +2$}.
    \end{cases}
\end{equation*}

\begin{claim}\label{claim:shift}
A solution to $\mathcal{I}$ with weighted sum of completion times $W$ can be transformed into a solution of $\mathcal{I}'$ with weighted sum of completion times at most $2W$.
Furthermore, any feasible schedule for $\mathcal{I}'$ also is a feasible schedule for $\mathcal{I}$.
\end{claim}
\begin{proof}
Let us define $s_i=\min\{t_i':\ t_i \leq t_i'\}$ for $i=1,\dots,q$.
Let $\pi$ be the solution for $\mathcal{I}$.
Then assigning all jobs that are assigned to time point $t_i$ to $s_i$ gives us a feasible solution to $\mathcal{I}'$.
By this change, the completion of any job is at most doubled (recall that each $t_i$ is assumed to be integer). 

Since the available amount of resources at each time in $\mathcal{I}'$ is at most as much as in $\mathcal{I}$, a feasible schedule for $\mathcal{I}'$ is also a feasible schedule for $\mathcal{I}$.
\end{proof}

\begin{claim}
There exists a polynomial time $(2+\varepsilon)$-approximation algorithm for constant $\varepsilon$ for all instances $\mathcal{I}$ where the resource arrival times are integer powers of $2$.
\end{claim}
\begin{proof}
We use the procedure that we described, i.e. we solve the instance for each of the time points using an $ \alpha$-approximation provided by \cite{kis2015approximability}, where $\alpha=1+\varepsilon$.
Let $\pi^\opt$ be an optimum solution and let $J^\opt_k$ be the set of jobs $j$ with $\pi^\opt(j)= k$.
We have 
\[ w(S_i)\leq \alpha \sum_{k=i}^q w(J^\opt_k)    \]
for $i=1,\dots,q$. Then we get
\begin{align*}
2\alpha \sum_{j\in J} w_j C_j^{\pi^\opt}
{}&{}=
\sum_{i=2}^q (2\alpha)\cdot 2^{i-2} w(J^\opt_i)\\
{}&{}\geq
\sum_{i=2}^q \alpha\cdot 2^{i-2} w(J^\opt_i)+\alpha \sum_{j=1}^\infty 2^{-j} 2^{i-2}w(J^\opt_i)\\
{}&{}\geq
\sum_{i=2}^q \alpha 2^{i-2}\sum_{k=i}^q w(J^\opt_k)\\
{}&{} \geq
\sum_{i=2}^q 2^{i-2} w(S_i)  ,
\end{align*}
thus the approximation ratio follows.
\end{proof}

The theorem follows from the two claims.
\end{proof}

\subsection{PTAS for constant \texorpdfstring{$q$}{q}} \label{sec:constq}

The aim of this section is to give a PTAS for the case when the number of resource arrival times is a constant.
The algorithm is a generalization of a well known PTAS for the knapsack problem, and will be used later as a subroutine in the PTAS for an arbitrary number of resource arrival times.
The idea is to choose a number $k\in\mathbb{Z}_+$, guess the $k$ heaviest jobs that are processed at each resource arrival time $t_i$, and then determine the remaining jobs that are scheduled at $t_i$ in a greedy manner.
Since we go over all possible sets containing at most $k$ jobs for each resource arrival time, there is an exponential dependence on the number $q$ of resource arrival times in the running time.

\begin{algorithm}[h!]
  \caption{PTAS for $1|rm=1,p_j=0|\sum C_jw_j$ when $q$ is a constant.}\label{alg:PTAS'}
  \begin{algorithmic}[1]
    \Statex \textbf{Input:} Jobs $\jobs$ with $|\jobs|=n$, resource requirements $a_j$, weights $w_j$, resource arrival times $t_1\le \ldots \le t_q$ and resource quantities $b_1, \dots b_q$.
    \Statex \textbf{Output:} A feasible schedule $\pi$. 
    \ForAll{subpartitions $S_1 \cup \dots \cup S_q \subseteq J $ with $|S_i|\leq k$ for $i>1$}  \label{st:1}
        \State Set $A=0$. \label{st:2}
        \State Set $W=0$. \label{st:3}
        \For{$i$ from $0$ to $q-2$} \label{st:4}
            \For{$j\in S_{q-i}$} \label{st:5}
                \State $\pi(j)=q-i$ \label{st:6}
                \State $A\leftarrow A+a_j$ \label{st:7}
            \EndFor
            \If{$|S_{q-i}|=k$} \label{st:8}
            \State $W\leftarrow \max\{W,\min\{w_j:\ j \in S_{q-i}\}\}$ \label{st:9}
    			\While{$A<B_{q-i}$} \label{st:10}
    				\If{there exists an unassigned job $j$ with $w_j\leq W$} \label{st:11}
    				    \State Let $j$ be an unassigned job with $w_j\leq W$ minimizing $w_j/a_j$. \label{st:12}
    				\State $\pi(j)=q-i$ \label{st:14}
    				\State $A\leftarrow A+a_j$ \label{st:15}
    				\Else
    				    \State \textbf{break} \label{st:13}
    				\EndIf
    			\EndWhile
            \EndIf 
        \EndFor
        \State For all remaining jobs set $\pi(j)=1$. \label{st:16}
    \EndFor
    \State Let $\pi$ be the best schedule found. \label{st:17}
    \State \textbf{return} $\pi$ \label{st:18}
  \end{algorithmic}
\end{algorithm}

\thmpreptas*
\begin{proof}
We claim that Algorithm~\ref{alg:PTAS'} satisfies the requirements of the theorem. Let $\pi^\opt$ be an optimal schedule and define $J^\opt_i=\{ j \in J:\ \pi^\opt(j)=i \}$.
Let $S^\opt_i$ be the set of the $k$ heaviest jobs in $J^\opt_i$ if $|J^\opt_i|\geq k$, otherwise let $S^\opt_i=J^\opt_i$.
Let $J_i=\{ j \in J:\ \pi(j)=i \}$ denote the set of jobs assigned to time $t_i$ in our solution.
In each iteration of the \emph{for} loop of Step~\ref{st:4}, let $j_i$ be the last job added to $J_i$ if such a job exists.

Assume that we are at the iteration of the algorithm when the subpartition $S^\opt_1\cup\dots\cup S^\opt_q$ is considered in Step~\ref{st:1}.
Let $W_{q-\ell}$ denote the value of $W$ at the end of the iteration of the for loop corresponding to $i=\ell$ in Step~\ref{st:4}.
By Steps~\ref{st:3} and \ref{st:9}, we have
\begin{equation*}
    W_{q-\ell}\leq \frac{1}{k}\sum_{i=\ell}^q\sum_{j \in J^\opt_{i}}w_j  .
\end{equation*}
As our algorithm always picks the most inefficient job, we also have 
\begin{equation*}
    \sum_{i=\ell}^q \sum_{j \in J_i\setminus \{j_i\}}w_j  \leq \sum_{i=\ell}^q \sum_{j\in J^\opt_{i}}w_j  ,
\end{equation*}
where $J_i\setminus\{j_i\}=J_i$ if $j_i$ is not defined for $i$. 

Combining these two observations, for $\ell=1,\dots,q$ we get
\begin{align*}
\sum_{i=\ell}^q \sum_{j \in J_i}w_j 
{}&{}=
\sum_{i=\ell}^q \sum_{j \in J_i\setminus \{j_i\}}w_j +\sum_{i=\ell}^q w_{j_i} \\
{}&{}\leq 
\sum_{i=\ell}^q \sum_{j\in J^\opt_i}w_j + (q-\ell+1)\cdot W_{\ell}\\
{}&{}\leq 
(1+\frac{q}{k})\sum_{i=\ell}^q \sum_{j\in J^\opt_{i}}w_j  ,
\end{align*}
where the first inequality follows from the fact that $w_{j_i}\leq W_i\leq W_\ell$ whenever $i\geq \ell$.
This proves that the schedule that we get is a $(1+\frac{q}{k})$-approximation.

We get a factor of $n^{qk}$ in the running time for guessing the sets $S_k$.
Assigning the remaining jobs can be done in linear time by ordering the jobs and using AVL-trees, thus we get an additional factor of $n$. In order to get a PTAS, we set $k=\frac{\varepsilon}{q}$, concluding the proof of the theorem.
\end{proof}

\subsection{PTAS for arbitrary \texorpdfstring{$q$}{q}}
\label{sec:ptas}

We turn to the proof of the main result of the paper. As in Section \ref{sec:4+eps}, we shift resource arrival times; here we use powers of $1+\varepsilon$, for a suitably small $\varepsilon$.

Let $\mathcal{I}$ be an instance of $1|rm=1, p_j=0 |\sum_j C_jw_j $. We assume that resource arrival times are integer, and that $t_1=0$, $t_2=1$.
We define a new instance $\mathcal{I}'$ of $1|rm=1, p_j=0 |\sum_j C_jw_j $ with shifted resource arrival times as follows. 
Set
\begin{equation*}
    t'_i=
    \begin{cases}
    0 &\quad \text{if $i=1$},\\
    (1+\varepsilon)^{i-2} &\quad \text{for $i=2,\dots,\lceil \log_{1+\varepsilon}(t_q) \rceil +2$},
    \end{cases}
\end{equation*}
and
\begin{equation*}
    b'_i=
    \begin{cases}
    b_i & \quad\text{if $i=1,2$},\\
    \sum[b_i:\ t_i\in ((1+\varepsilon)^{i-3}, (1+\varepsilon)^{i-2}]\ ] &\quad \text{for $i=3,\dots,\lceil \log_{1+\varepsilon}(t_q) \rceil +2$}.
    \end{cases}
\end{equation*}

The proof of the following claim is the same as that of Claim \ref{claim:shift}.

\begin{claim}
A solution to $\mathcal{I}$ with weighted sum of completion times $W$ can be transformed into a solution of $\mathcal{I}'$ with weighted sum of completion times at most $(1+\varepsilon)W$.
Furthermore, any feasible schedule for $\mathcal{I}'$ also is a feasible schedule for $\mathcal{I}$. \qed
\end{claim}

Due to the claim, we may assume that the positive arrival times are powers of $1+\varepsilon$. For convenience of notation, in this section we will assume that the largest arrival time is 1, and arrival times are indexed in decreasing order, starting with $t_0=1$. That is, $t_i=(1+\varepsilon)^{-i}$ ($i=0,\dots,q-2$), and $t_{q-1}=0$.
We will also assume that for a given constant $r$, $b_{q-r-1}=\dots=b_{q-2}=0$. This can be achieved by adding $r$ dummy arrival times.

\thmptas*
\begin{proof}
Let us fix an even integer $r$ and $\varepsilon>0$; we will later assume that $r$ is very large compared to $\varepsilon^{-1}$.
We assume that resource arrival times are as described above, and are indexed in decreasing order.  

In the algorithm, we fix jobs at progressively decreasing arrival times, by using the PTAS of the previous section for $r+1$ arrival times on different instances except for the first step, when we may use the PTAS for less than $r+1$ arrival times. 
We will run our algorithm $r/2$ times with slight modifications, and pick the best result.
Each run is characterized by a parameter $\ell \in \{1,\dots,r/2\}$.

In the first step, we consider arrival times $t_0,t_1,\dots,t_{r/2+\ell-1},0$.
We move the resources arriving before $t_{r/2+\ell-1}$ to 0, and use the PTAS for $r/2+\ell+1$ arrival times on this instance.
We fix the jobs that are scheduled at arrival times $t_0,t_1,\dots, t_{\ell-1}$.

Consider now the $j$th step for some $j\geq 2$. 
Define $s=(j-2)r/2+\ell$ and consider arrival times $t_s,t_{s+1},\dots,t_{s+r-1},0$.
Move the resources arriving before $t_{s+r-1}$ to 0, and decrease $b_s,b_{s+1},\dots$ in this order as needed, so that the total requirement of unfixed jobs equals the total resource.
Use the PTAS for $r+1$ arrival times on this instance.
Fix the jobs that are scheduled at arrival times $t_s,t_{s+1},\dots, t_{s+r/2-1}$.

The algorithm runs while $s+r-1 \leq q-2$, i.e., $jr/2+\ell \leq q-1$. Since the smallest $r$ arrival times (except for $0$) are dummy arrival times, the algorithm considers all resource arrivals.

The schedule given by the algorithm is clearly feasible, because when jobs at $t_i$ are fixed, the total resource requirement of jobs starting no earlier than $t_i$ is at least the total amount of resource arriving no earlier than $t_i$.
To analyze the approximation ratio, we introduce the following notation:
$W_i$ is the total weight that the algorithm schedules at $t_i$;
$W'_i$ is the weight that the algorithm temporarily schedules at $t_i$ when $i$ is in the interval $[t_{s+r/2},t_{s+r-1}]$ (or, in the first step, in the interval $[t_{\ell},t_{\ell+r/2-1}]$);
$W^*_i$ is the total weight scheduled at $t_i$ in the optimal solution.

Since we use the PTAS for $r/2+\ell+1$ arrival times in the first step, we have
\begin{equation*}
  \sum_{i=0}^{\ell-1}(1+\varepsilon)^{-i}W_i + \sum_{i=\ell}^{\ell+r/2-1}(1+\varepsilon)^{-i}W'_i \leq (1+\varepsilon)\sum_{i=0}^{\ell+r/2-1}(1+\varepsilon)^{-i}W^*_i,
\end{equation*}
as the right-hand side is $(1+\varepsilon)$ times the objective value of the feasible solution obtained from the optimal solution by moving jobs arriving before $t_{\ell+r/2-1}$ to~0.

For $s=jr/2+\ell$, we compare the output of the PTAS with a different feasible solution: we schedule total weight $W'_i$ at $t_i$ for $i=s,s+1,\dots,s+r/2-1$, total weight $W^*_i$ at $t_i$ for $i=s+r/2+1,\dots,s+r-1$, and at $t_{s+r/2}$ we schedule all jobs that are no earlier than $t_{s+r/2}$ in the optimal schedule but are no later than $t_{s+r/2}$ in the PTAS schedule.
We get the inequality
\begin{multline*}
  \sum_{i=jr/2+\ell}^{(j+1)r/2+\ell-1}(1+\varepsilon)^{-i}W_i + \sum_{i=(j+1)r/2+\ell}^{(j+2)r/2+\ell-1}(1+\varepsilon)^{-i}W'_i\\
  \leq (1+\varepsilon)\left(\sum_{i=jr/2+\ell}^{(j+1)r/2+\ell-1}(1+\varepsilon)^{-i}W'_i +\sum_{i=(j+1)r/2+\ell}^{(j+2)r/2+\ell-1}(1+\varepsilon)^{-i}W^*_i \right.\\ \left. + (1+\varepsilon)^{-(j+1)r/2-\ell}\sum_{i=0}^{(j+1)r/2+\ell-1}W^*_i\right)  .
\end{multline*}
The sum of these inequalities gives
\begin{align}
  \sum_{i=0}^{q-2}(1+\varepsilon)^{-i}W_i
  {}&{}\leq
  \varepsilon \sum_{i=\ell}^{q-2}(1+\varepsilon)^{-i}W'_i + (1+\varepsilon)\sum_{i=0}^{q-2}(1+\varepsilon)^{-i}W^*_i\nonumber\\ 
  {}&{}~~~+ (1+\varepsilon)\sum_{i=0}^{q-2}\left(\sum_{j:jr/2+\ell>i}(1-\varepsilon)^{-(jr/2+\ell)}\right)W^*_i  . \label{eq:bound1}
\end{align}
To bound the first term on the right hand side of \eqref{eq:bound1}, first we observe that
\begin{equation*}
  \sum_{i=\ell}^{r/2+\ell-1}(1+\varepsilon)^{-i}W'_i \leq (1+\varepsilon)\sum_{i=0}^{r/2+\ell-1}(1+\varepsilon)^{-i}W^*_i, 
\end{equation*}
because the left side is at most the value of the PTAS in the first step, while the right side is $(1+\varepsilon)$ times the value of a feasible solution.
Similarly,
\begin{eqnarray*}
  &\displaystyle\sum_{i=(j+1)r/2+\ell}^{(j+2)r/2+\ell-1}(1+\varepsilon)^{-i}W'_i \leq&\\ &(1+\varepsilon)\left(\displaystyle\sum_{i=jr/2+\ell}^{(j+2)r/2+\ell-1}(1+\varepsilon)^{-i}W^*_i  +(1+\varepsilon)^{-jr/2-\ell}\displaystyle\sum_{i=0}^{jr/2+\ell-1}W^*_i\right), &
\end{eqnarray*}
because the left side is at most the value of the PTAS in the $(j+1)$-th step, and the right side is $(1+\varepsilon)$ times the value of the following feasible solution: take the optimal solution, move jobs scheduled before $t_{(j+2)r/2+\ell-1}$ to 0, and move jobs scheduled after $t_{jr/2+\ell}$ to $t_{jr/2+\ell}$.
Adding these inequalities, we get
\begin{eqnarray*}
  &\varepsilon \displaystyle\sum_{i=\ell}^{q-2}(1+\varepsilon)^{-i}W'_i \leq&\\ &\varepsilon(1+\varepsilon)\left(2\displaystyle\sum_{i=0}^{q-2}(1+\varepsilon)^{-i}W^*_i + \displaystyle\sum_{i=0}^{q-2}\left(\sum_{j:jr/2+\ell>i}(1+\varepsilon)^{-jr/2-\ell}\right) W^*_i\right)\leq &\\
  &\varepsilon(1+\varepsilon)\left(2\displaystyle\sum_{i=0}^{q-2}(1+\varepsilon)^{-i}W^*_i + \displaystyle\sum_{i=0}^{q-2}\left(\displaystyle\sum_{j=0}^{\infty}(1+\varepsilon)^{-jr/2-1}\right)(1+\varepsilon)^{-i} W^*_i\right)=&\\
  & \varepsilon(1+\varepsilon)\left(2\displaystyle\sum_{i=0}^{q-2}(1+\varepsilon)^{-i}W^*_i + \frac{(1+\varepsilon)^{r/2-1}}{(1+\varepsilon)^{r/2}-1}\displaystyle\sum_{i=0}^{q-2}(1+\varepsilon)^{-i} W^*_i\right)=&\\
  &\varepsilon \left(2(1+\varepsilon)+\frac{(1+\varepsilon)^{r/2}}{(1+\varepsilon)^{r/2}-1}\right)\displaystyle\sum_{i=0}^{q-2}(1+\varepsilon)^{-i} W^*_i   .&
\end{eqnarray*}
The last expression is at most $4\varepsilon$ times the optimum value if $r$ is large enough.

The last term of the right side of \eqref{eq:bound1} is too large to get a bound that proves a PTAS.
However, we can bound the \emph{average} of these terms for different values of $\ell$.
The average is
\begin{eqnarray*}
  &(1+\varepsilon)\frac{2}{r}\displaystyle\sum_{\ell=1}^{r/2}\displaystyle\sum_{i=0}^{q-2}\left(\displaystyle\sum_{j:jr/2+\ell>i}(1-\varepsilon)^{-(jr/2+\ell)}\right)W^*_i \leq &\\
  & (1+\varepsilon)\frac{2}{r} \displaystyle\sum_{i=0}^{q-2}\left(\displaystyle\sum_{j=1}^{\infty}(1+\varepsilon)^{-j} \right) (1-\varepsilon)^{-i} W^*_i= (1+\varepsilon)\frac{2}{r \varepsilon} \displaystyle\sum_{i=0}^{q-2} (1-\varepsilon)^{-i} W^*_i, &
\end{eqnarray*}
which is at most $\varepsilon$ times the optimum if $r$ is large enough.
To summarize we obtained that for large enough $r$, the average objective value of our algorithm for $\ell=1,2,\dots,r/2$ is upper bounded by
\begin{equation*}
  4\varepsilon \sum_{i=0}^{q-2}(1+\varepsilon)^{-i} W^*_i+(1+\varepsilon)\sum_{i=0}^{q-2}(1+\varepsilon)^{-i}W^*_i+\varepsilon \sum_{i=0}^{q-2}(1+\varepsilon)^{-i} W^*_i =(1+6\varepsilon) \sum_{i=0}^{q-2}(1+\varepsilon)^{-i} W^*_i,
\end{equation*}
which is $(1+6\varepsilon)$ times the objective value of the optimal solution.
This proves that the algorithm that chooses the best of the $r/2$ runs is a PTAS.
\end{proof}

\subsection{Undetermined resource arrival times}

A question one might ask is the following:
Given $\alpha>1$ and an instance of $1|rm=1, p_j=0 |\sum_j C_jw_j $, is there a schedule $\pi$ such that the set $S_i=\{ j:\ \pi(j)\geq i \}$ is an $\alpha$-approximation to the problem of finding a minimum weight job set $S\subseteq J$ consuming at least $B_i$ resources for $i=1,\dots,q$? 
The motivation is that such a solution would also give an $\alpha$-approximation for the instance.
The greedy algorithm of Section~\ref{sec:greedy} shows that the answer to the question is yes if $\alpha\geq 6$.

Let us consider now a variant of the problem where the arriving resource quantities are given, but the resource arrival times are not known in advance. 
Then the smallest $\alpha$ for which the answer is `YES' in the above problem is the best approximation ratio we can achieve.

\thmunknown*
\begin{proof}
Our approximation algorithm is based on the following claim.

\begin{claim}\label{cl:unknown1}
There exists a schedule $\pi$ such that for each $i$ the set $S_i=\{ j:\ \pi(j)\geq i \}$ is a $4$-approximation for the problem of finding a minimum weight job set $S\subseteq J$ consuming at least $B_i$ resources.
\end{claim}
\begin{proof}
For $i=1,\dots,q$, let $J_i$ be an optimal solution to the problem of finding a minimum weight job set $S\subseteq J$ consuming at least $B_i$ resources. Define $f(i)= \min\{k:\ w(J_k)\leq 2w(J_i)\}$ for $i=2,\dots,q$ and let us consider the following procedure.

\begin{algorithm}[h!]
  \caption{Subroutine for $(4+\varepsilon)$-approximation to $1|rm=1, p_j=0, \text{$t_i$ unkown}|\sum C_jw_j$.}\label{alg:unknown}
  \begin{algorithmic}[1]
    \Statex \textbf{Input:} Jobs $\jobs$ with $|\jobs|=n$, resource requirements $a_j$, weights $w_j$, resource quantities $b_1, \dots b_q$.
    \Statex \textbf{Output:} A feasible schedule $\pi$. 
    \State Set $i=q$.
    \While{$i \geq 1$} \label{st:while}
        \State Set $S_{i+1}=\{ j:\ \pi(j)> i \}$.
        \State Set $\pi(j)=i$ for $j\in J_{f(i)}\setminus S_{i+1}$.
        \State $i\leftarrow f(i)-1$
    \EndWhile
    \State \textbf{return} $\pi$
  \end{algorithmic}
\end{algorithm}

It is not difficult to see that $\pi$ fulfills the resource requirements.
We prove by induction that $w(S_i)\leq 4w(J_i)$ holds for $i=1,\dots,q$.
As $S_q=J_{f(q)}$, the inequality $w(S_q)\leq 4w(J_q)$ clearly holds.
Assume now that $i\leq q-1$. If no jobs are assigned to $t_i$, then $w(S_i)=w(S_{i+1})\leq 4w(J_{i+1}) \leq 4w(J_i)$.
Otherwise $i=f(i')-1$, where $i'$ is the index considered in the previous iteration of the while loop in Step~\ref{st:while}. Observe that no jobs are assigned to time points between the $i$th and the $i'$th ones. By induction, we get
\begin{equation*}
    w(S_{i})=w(J_{f(i)}\cup S_{i'}) \leq w(J_{f(i)})+w(S_{i'})\leq 2w(J_{i}) + 4w(J_{i'}) \leq 4w(J_{i})  .
\end{equation*}
Here the second inequality holds by induction and by the definition of $f$, while the last inequality follows from the fact that $i<f(i')$ which implies $w(J_{i}) > 2w(J_{i'}) $.
\end{proof}

Algorithm~\ref{alg:unknown} provides a $(4+\varepsilon)$-approximation for $1|rm=1, p_j=0, \text{$t_i$ unkown}|\sum C_jw_j$ which has running time polynomial in the input size and $\frac{1}{\varepsilon}$ as follows.
Using either an FPTAS for the knapsack problem or the FPTAS of Kis~\cite{kis2015approximability}, we determine an approximation the sets $J_i$.
Then we apply Algorithm~\ref{alg:unknown} to schedule the jobs. This concludes the proof of the first part of the theorem.

The following set of instances shows that $\alpha$ is at least 4. We are given $(n-1)m$ jobs denoted by $1, 2, \dots, (n-1)m$ with weights $w_i=n-\frac{i}{m}$ and $a_i=2^{-i}$.
Furthermore, we have $(n-1)m$ resource arrival times that are unknown.
The resource quantities are given by $b_q=2^{-q}$ and $b_i=2^{-i}-B_{i+1}=2^{-i-1}$ for $i<q$.
In order to fulfill the resource requirements, the set of jobs scheduled at or after time $t_i$ has to contain at least one of the jobs $j\leq i$.
Observe that the optimal solution of finding a minimum weight job set $J_i$ consuming at least $B_i$ resources consists of the single job $i$.

Let $j_1$ be the job processed at time $t_1$. Now we create a sequence starting with $j_1$.
Since the jobs that are greater or equal than $j_1$ have total resource requirements less than $B_{j_1-1}$, we have to schedule at least one job $j_2<j_1$ at or after $t_{j_1-1}$.
This argument can be iterated to find a job $j_3 < j_2$ which is scheduled at or after $t_{j_2-1}$, and so on.

\begin{claim} \label{cl:inap}
For any $\beta< 4$, there exist $n$ and $m$ such that for the sequence $j_1,j_2,\dots$ constructed as above, there is some $i$ with
\[ \beta w_{j_i-1}< \sum_{k=1}^{i+1}w_{j_k}. \]
\end{claim}
\begin{proof}
Suppose to the contrary that there is no triple $n,m, i$ satisfying the requirements of the claim.
Then $\beta w_{j_i-1}\geq \sum_{k=1}^{i+1}w_{j_k}$ for all $n,m, i$.
If $m$ is large enough, then $w_{j_i}$ and $w_{j_i-1}$ are very close to each other, i.e. if $\frac{1}{m}\leq \varepsilon_m$ for some $\varepsilon_m$ then $w_{j_i-1} \leq w_{j_i}+\varepsilon_m$.
Thus we get $\beta w_{j_{i}-1}\leq \beta  w_{j_i}+\beta\varepsilon_m$.
By increasing $n$, the length of our sequence increases as well. Indeed, by the indirect assumption, we have 
\[ w_{j_{i+1}} \leq \sum_{k=1}^{i+1}w_{j_k} \leq \beta w_{j_{i}-1} <  4w_{j_{i}-1} < 4w_{j_{i}}+4\varepsilon_m.\]
Since we also have $w_{j_1}\leq 4$ and $w_j\geq 1$ for all jobs $j$, this implies that the number of elements in the sequence is at least $\log_5 n$.
Let us define $z_i=w_{j_i}$ and $z_i'=\sum_{k=1}^{i-1}z_k$.
By the indirect assumption, $\beta (z_{i}+\frac{1}{m}) \geq z_{i+1}+z_{i+1}'$, thus we have
\[ (\beta-1)z_i-z_i'+4\varepsilon_m \geq z_{i+1}  . \]
We omit the $\varepsilon_m$ for now to keep the reasoning simple.
We claim that $\frac{z_{i+1}'}{z_{i+1}}\geq \gamma \frac{z_i'}{z_i}$ for some $\gamma>1$ which only depends on $\beta$.
We have $(\beta-1)z_{i}-z_i'-z_{i+1}\geq 0$, and thus we get $3z_i>z_i'$ assuming $\beta<4$.
Observe that $\frac{z_i'}{z_i}$ is minimal if $(\beta-1)z_i-z_i' = z_{i+1}$.
By transforming the equation
\[ \gamma \frac{z_{i}'}{z_{i}}= \frac{z_{i+1}'}{z_{i+1}}=\frac{z_i+z_i'}{(\beta-1)z_i-z_i'}, \]
we get
\begin{align}\label{eqgamma}
\gamma=\frac{z_{i}^2+z_{i}z_i'}{(\beta-1)z_iz_i'-z_i'^2} .
\end{align}
Using $(z_i-z_i')^2\geq 0$, $(\beta-2)<2$ and $z_iz_i'>0$ we get
\begin{align*}
z_{i}^2+z_i'^2 > 2z_iz_i'
\end{align*}
or equivalently
\begin{align*}
z_{i}^2+z_{i}z_i' > 3z_iz_i'-z_i'^2 
\geq 3z_iz_i'-\frac{3}{\beta-1}z_i'^2
=\frac{3}{\beta-1}((\beta-1)z_iz_i'-z_i'^2).
\end{align*}
Thus $\gamma>\frac{3}{\beta-1}$.
It follows that there exists some $N$ such that $\frac{z_N'}{z_N}>3$ and thus we get
\[ \beta z_N < 4 z_N \leq z_N+z_N',   \]
contradicting the indirect assumption.
\end{proof}

By Claim~\ref{cl:inap}, there exists a time point $t_i$ with $\sum_{j:\pi(j)\geq i} w_j>\beta B_i$. By setting $t_{i'}=0$ (or very close to $0$) for $i'<i$  and $t_{i'}=1$ (or very close to 1) for $i'\geq i$, the schedule can only be a $(\beta- \varepsilon)$-approximation if we do not know the resource arrival times in advance.
\end{proof}

\section*{Acknowledgement}

The authors are grateful to Erika B\'erczi-Kov\'acs and to Matthias Mnich for the helpful discussions. Krist\'of B\'erczi was supported by the J\'anos Bolyai Research Fellowship of the Hungarian Academy of Sciences and by the ÚNKP-19-4 New National Excellence Program of the Ministry for Innovation and Technology. Project no. NKFI-128673 has been implemented with the support provided from the National Research, Development and Innovation Fund of Hungary, financed under the FK\_18 funding scheme. Tamás Király was supported by the Hungarian National Research, Development and Innovation Office -- NKFIH,
grant number K120254. This research was supported by Thematic Excellence Programme, Industry and Digitization Subprogramme, NRDI Office, 2019.

%
%
%
%
\bibliographystyle{abbrv}
\bibliography{references}

\end{document}